\documentclass{article}

\usepackage{amsmath}
\usepackage{amssymb}
\usepackage{amsthm}
\usepackage{mathtools} 
\usepackage{bm}    
\usepackage{color}
\usepackage{enumitem}  
\newtheorem{theorem}{Theorem}
\newtheorem{proposition}{Proposition}

\newtheorem{lemma}{Lemma}
\usepackage{cite} 

\usepackage{spconf}
\usepackage{graphicx}
\usepackage{hyperref}

\title{Stable  Filters for Generative Modeling of Graph Signals}

\name{Martin Schmidt and Gonzalo Mateos\thanks{Work supported in part by the NSF under award ECCS 2231036.}}
\address{University of Rochester, Rochester, NY, USA}

\begin{document}
\ninept

\maketitle

\begin{abstract}
Generating signals on graphs requires permutation-equivariant models that exhibit stability with respect to relative structural perturbations. While recent graph-aware Schrödinger bridge models incorporate topology information directly into their reference dynamics, it is unclear how perturbations of the graph propagate through these dynamics and affect the resulting generated distributions. In this paper, we analyze the structural stability of graph-aware continuous-time generative models whose drift combines a graph filter with a learned graph neural network. We derive explicit Wasserstein stability bounds that quantify the effect of relative graph perturbations on the generated distributions. Motivated by these bounds, we introduce a principled framework for designing stable graph filters that preserve the smoothing behavior of graph heat diffusion, while boosting structural stability. Experiments on synthetic and fMRI signals show our stable filters enhance structural robustness while matching or exceeding the generative quality of the heat equation baseline.
\end{abstract}


\section{Introduction}

Mapping one probability distribution to another is a fundamental problem in generative modeling. While these models have seen remarkable success in Euclidean spaces, there is a growing need to transport distributions on non-Euclidean settings. In particular, graphs provide a natural representation for these irregular domains, where vertex-supported signals can describe, e.g., neural activity over brain connectomes or traffic flows on transportation networks~\cite{huang2018graph,schmidt2025connectome,yu2017spatio}. Learning to faithfully generate graph signals has broad applicability, ranging from synthetic data augmentation and implicit priors to realistic system simulations across domains like recommender systems, financial forecasting, and wireless networks \cite{zhu2024graphsignaldiffusionmodel, uslu2026generative}.

Among the many stochastic processes that match a given pair of endpoint distributions, the dynamic Schr\"odinger bridge problem provides a principled solution that favors the process closest to a prescribed reference dynamics~\cite{de2021diffusion, chen2022likelihood, liu2024generalized}. In the classical formulation, the reference is typically chosen as a standard Wiener process. However, when transporting \emph{distributions of graph signals}, it is prudent to incorporate the underlying structure directly into the generative model by replacing the Wiener process with graph-aware reference dynamics~\cite{yang2025topological, wyrwal2026topological, rozada2025graphawarediffusionsignalgeneration}. This promotes trajectories that are consistent with the graph topology, and the corresponding dynamics can be learned from samples using bridge and flow matching techniques~\cite{shi2023diffusion, lipman2022flow}.

Crucially and by design, these generative models explicitly depend on the graph that parameterizes both the reference dynamics and the learned vector field. In practice, the observed graph is only a partial reflection of an underlying complex system~\cite[Ch. 7]{kolaczyk2009book}, typically subject to missing edges or estimation errors, leading to an imperfect graph representation~\cite{SI_SPMAG, dong_2019_learning}. Consequently, unavoidable perturbations of the graph can steer the entire generative trajectory and, ultimately, the distribution of the generated signals.

Recent efforts have started to examine the robustness properties of continuous-time generative models on graphs. Specifically, prior work established the structural stability of generative processes governed solely by a learned graph neural network (GNN) vector field~\cite{schmidt_2026}. However, the analysis therein does not extend to frameworks that involve graph-aware reference dynamics, which introduce an additional graph-dependent drift term governed by a graph filter. Since this filter is a fixed design choice (e.g., graph heat diffusion) rather than a learned component, its inclusion not only alters the stability of the overall dynamics, but also raises the critical question of how to design stable graph filters for these generative models.

Our contributions address this foundational gap. First, in Section \ref{sec:stability_properties} we characterize the structural stability of graph-aware continuous-time models by deriving bounds on the Wasserstein distance, quantifying the error between the generated distributions caused by graph perturbations. Leveraging these theoretical guarantees, in Section \ref{sec:stable_filters} we introduce a principled optimization framework for reference graph filter design. Specifically, our approach synthesizes filters that respect the induced smoothing of heat diffusion on graphs without sacrificing robustness to structural perturbations. Reproducible tests on synthetic and fMRI signals demonstrate that our optimal filter design exhibits enhanced stability to graph perturbations, while matching or exceeding the generative quality of the heat equation baseline (Section \ref{sec:experiments}). Concluding remarks are given in Section \ref{sec:conclusion}.

\section{Preliminaries}

\subsection{Graph Signal Processing}

Let $\mathcal{G}= (\mathcal{V}, \mathcal{E}, \mathcal{W})$ be a known undirected graph with $N$ nodes, edge set $\mathcal{E}\subseteq \mathcal{V}\times \mathcal{V}$, and edge weight map $\mathcal{W}: \mathcal{E}\to \mathbb{R}$. Its structure is represented by a graph shift operator (GSO) $\mathbf{L}\in\mathbb{R}^{N\times N}$, such as the graph Laplacian. A graph signal $x:\mathcal{V}\to\mathbb{R}$ is represented by $\mathbf{x}\in\mathbb{R}^N$, with $x_i$ denoting its value at node $i$. Assuming the eigendecomposition $\mathbf{L}=\mathbf{V}\mathbf{\Lambda}\mathbf{V}^\top$, with eigenvalues $\{\lambda_i\}_{i=1}^N$, the graph Fourier transform is $\hat{\mathbf{x}}=\mathbf{V}^\top\mathbf{x}$ \cite{ortega2018graph}. A graph convolutional filter of order $P$ is defined as $\mathbf{F}(\mathbf{L})\coloneqq\sum_{p=0}^{P}\theta_p\mathbf{L}^p$, with coefficients $\{\theta_p\}$~\cite{sandryhaila2014tsp,isufi2024tsp}, and corresponding frequency response $f(\lambda)\coloneqq\sum_{p=0}^{P}\theta_p\lambda^p$. A GNN extends this construction by cascading graph filters with pointwise, normalized Lipschitz nonlinearities $\sigma(\cdot)$ \cite{Gama_2019, kipf2016semi, defferrard2016convolutional}. For an input $\mathbf{x}$, we denote the output of a GNN with learnable parameters $\boldsymbol{\theta}$ and fixed graph support $\mathbf{L}$ by $u_{\boldsymbol{\theta}}(\mathbf{x};\mathbf{L})$.

\subsection{Stability of Graph Filters and GNNs}
\label{sec: stability-GNN}

Let $\mathcal{P}$ denote the set of $N \times N$ permutation matrices. A fundamental property of both graph filters and GNNs is \emph{permutation equivariance}~\cite{Gama_2020}. For any permutation $\mathbf{P} \in \mathcal{P}$, applying the operator to a permuted signal $\mathbf{P}^\top \mathbf{x}$ on a similarly permuted graph $\mathbf{P}^\top \mathbf{L} \mathbf{P}$ yields a permuted output, meaning $u_{\boldsymbol{\theta}}(\mathbf{P}^\top \mathbf{x}; \mathbf{P}^\top \mathbf{L} \mathbf{P}) = \mathbf{P}^\top u_{\boldsymbol{\theta}}(\mathbf{x}; \mathbf{L})$. 

In practice, the observed graph structure is often noisy. To analyze the robustness of graph filters and GNNs to structural errors, we adopt a well-documented relative perturbation model~\cite{Gama_2020, luana_2021}. Given a nominal GSO $\mathbf{L}$ and a perturbed one $\tilde{\mathbf{L}}$, we evaluate the symmetric relative error matrix $\mathbf{E} \in \mathcal{S}:=\{\mathbf{S}\in\mathbb{R}^{N\times N}: \mathbf{S}=\mathbf{S}^\top\}$ at the permutation $\mathbf{P}_0 \in \mathcal{P}$ that minimizes the error's spectral norm, i.e.,
\begin{equation}
\label{eq: perturbation-model}
\begin{aligned}
\{\mathbf{E}^\star, \mathbf{P}_0\} &= \arg\min_{\mathbf{E}\in\mathcal{S}, \mathbf{P}\in\mathcal{P}} \|\mathbf{E}\|_2 \\
\text{s.t.} \quad \mathbf{P}^{\top}\tilde{\mathbf{L}}\mathbf{P} &= \mathbf{L} + \frac{1}{2}(\mathbf{E}\mathbf{L}+\mathbf{L}\mathbf{E}).
\end{aligned}
\end{equation}
Under this perturbation model, both graph filters and GNNs are \emph{structurally stable}. Specifically, for a relative perturbation bounded by $\|\mathbf{E}^\star\|_2 \leq \varepsilon$, the pointwise deviation of the output is bounded by
\begin{equation*}
\| \mathbf{P}_0^\top u_{\boldsymbol{\theta}}(\mathbf{x}; \tilde{\mathbf{L}}) - u_{\boldsymbol{\theta}}(\mathbf{P}_0^\top \mathbf{x}; \mathbf{L}) \| \leq \left(\Gamma \varepsilon + \mathcal{O}(\varepsilon^2)\right) \|\mathbf{x}\|,
\end{equation*}
for all $\mathbf{x} \in \mathbb{R}^N$. The structural stability constant $\Gamma > 0$ dictates the sensitivity of the architecture to structural noise. For a graph filter, $\Gamma$ depends primarily on the integral Lipschitz constant of its frequency response; see~\cite[Thm. 2]{Gama_2020}. For a GNN, however, $\Gamma$ is strictly larger as it compounds with the network's depth, width, and the uniform spectral bounds of its intermediate filters~\cite[Thm. 4]{Gama_2020}. We build on these foundations to analyze the propagation of structural errors through the continuous-time dynamics of our generative model.

\subsection{Topological Schr\"odinger Bridge}

Let $p_0$ and $p_1$ be two probability distributions on $\mathbb{R}^N$, and let $\mathfrak{M}$ denote the set of probability measures on the path space $\mathcal{C}([0,1],\mathbb{R}^N)$. Given a reference path measure $\mathbb{Q}$, the Schr\"odinger bridge problem \cite{leonard2013survey} seeks (KL stands for Kullback-Leibler divergence)
\begin{equation}
\label{eq: sb}
    \mathbb{P}^{\star}
    \in
    \arg\min_{\mathbb{P}\in \mathfrak{M}}
    \mathrm{KL}(\mathbb{P}\|\mathbb{Q})
    \quad \text{s.t.} \quad
    \mathbb{P}_0=p_0,\ \mathbb{P}_1=p_1.
\end{equation}
Thus, the choice of $\mathbb{Q}$ determines the reference dynamics with respect to which the interpolation between $p_0$ and $p_1$ is optimized. 

For graph-supported data, this reference process can be chosen to explicitly incorporate the underlying topology \cite{yang2025topological}. Let $\mathbf{L}$ denote a GSO and consider the graph-aware reference diffusion
\begin{equation}
\label{eq: reference-SDE}
    d\mathbf{y}_t
    =
    \left[
        \mathbf{H}_t(\mathbf{L})\mathbf{y}_t
        +
        \boldsymbol{\alpha}_t
    \right]dt
    +
    \mathbf{\Sigma}_t\,d\mathbf{w}_t,
\end{equation}
where $\mathbf{w}_t$ is a standard Brownian motion in $\mathbb{R}^N$, $\mathbf{\Sigma}_t\in\mathbb{R}^{N\times N}$ controls the diffusion, $\boldsymbol{\alpha}_t \in \mathbb{R}^N$ is a bias term and $\mathbf{H}_t(\mathbf{L})\coloneqq \sum_{k=0}^{K} h_k(t)\mathbf{L}^k$ is a time-dependent graph filter of order $K$. Then $\mathbb{Q}$ is the path measure law induced by \eqref{eq: reference-SDE}. Substituting $\mathbb{Q}$ into \eqref{eq: sb} yields a Schr\"odinger bridge whose optimal trajectories are defined relative to the graph-aware dynamics. For general endpoint distributions, the drift of the resulting bridge is not available in closed form. Bridge and flow matching methods learn the required correction from samples by regressing vector fields associated with conditional reference bridges~\cite{yang2025topological, wyrwal2026topological}. The resulting dynamics can be written as
\begin{equation}
\label{eq: generative-SDE}
    d\mathbf{x}_t
    =
    b_t(\mathbf{x}_t;\mathbf{L})dt
    +
    \mathbf{\Sigma}_t\,d\mathbf{w}_t,
\end{equation}
where
\begin{equation}
\label{eq: generative-drift}
    b_t(\mathbf{x};\mathbf{L})
    \coloneqq
    \mathbf{H}_t(\mathbf{L})\mathbf{x}
    +
    \boldsymbol{\alpha}_t
    +
    u_t^{\boldsymbol{\theta}}(\mathbf{x};\mathbf{L}).
\end{equation}
Here, $u_t^{\boldsymbol{\theta}}(\cdot;\mathbf{L})$ denotes a learned vector field that corrects the reference drift so that the law induced by \eqref{eq: generative-SDE} solves the Schr\"odinger bridge problem \eqref{eq: sb}. Specifically, we let
$u_t^{\boldsymbol{\theta}}(\mathbf{x};\mathbf{L}) \coloneqq u_{\boldsymbol{\theta}}(g(\mathbf{x},t);\mathbf{L})$, where $u_{\boldsymbol{\theta}}(\cdot;\mathbf{L})$ is a GNN and $g(\mathbf{x},t)\in\mathbb{R}^N$ is a continuous, permutation-equivariant conditioning map that incorporates the temporal dependence; see \cite[Sec. III]{schmidt_2026} for a detailed discussion.

\section{Stability Properties}\label{sec:stability_properties}

\subsection{Problem Formulation}

We consider the graph-aware generative dynamics \eqref{eq: generative-SDE}--\eqref{eq: generative-drift} and study their stability to perturbations of the underlying graph $\mathcal{G}$. In particular, let $\mathbf{L}$ denote the nominal GSO and $\tilde{\mathbf{L}}$ be a perturbed GSO satisfying the relative perturbation model \eqref{eq: perturbation-model}. Let $\Phi_t(\mathbf{x}_0;\mathbf{L})$ denote the solution of \eqref{eq: generative-SDE} at time $t$ with initial condition $\mathbf{x}_0\sim p_0$. Because node indexing is arbitrary, we compare the nominal and perturbed dynamics modulo permutations. For a permutation $\mathbf{P}\in\mathcal{P}$, define
\begin{equation}
\label{eq: dynamics}
\Phi_t(\mathbf{P}^\top\mathbf{x}_0;\mathbf{L})
        \sim \pi_t^{\mathbf{P}},\quad
\mathbf{P}^\top\Phi_t(\mathbf{x}_0;\tilde{\mathbf{L}})\sim \tilde{\pi}_t^{\mathbf{P}}.
\end{equation}
Our goal is to characterize the Wasserstein distance between these distributions as a function of the structural perturbation
$\|\mathbf{E}^{\star}\|_2\leq\varepsilon$. 

\subsection{Permutation Equivariance}

To ensure the generative dynamics respect the graph's structural symmetries, the stochastic differential equation (SDE) in \eqref{eq: generative-SDE} must be permutation equivariant. The graph filter $\mathbf{H}_t(\mathbf{L})$ and the GNN vector field $u_t^{\boldsymbol{\theta}}(\cdot;\mathbf{L})$ are equivariant by design (cf. Section \ref{sec: stability-GNN}). We further assume that for any $\mathbf{P} \in \mathcal{P}$, the bias and diffusion satisfy $\mathbf{P}^\top\boldsymbol{\alpha}_t = \boldsymbol{\alpha}_t$ (e.g., a uniform drift $\boldsymbol{\alpha}_t = c \mathbf{1}$) and $\mathbf{P}^\top\mathbf{\Sigma}_t\mathbf{P} = \mathbf{\Sigma}_t$ (e.g., isotropic noise $\mathbf{\Sigma}_t = \sigma_t \mathbf{I}$) for all $t \in [0,1]$. Under these conditions, node relabeling commutes with the stochastic dynamics\footnote{Selected detailed proof arguments are omitted due to limited space.}.

\begin{proposition}
\label{prop: dynamics-equivariance}
For any initial distribution $\mathbf{x}_0 \sim p_0$, the probability laws of the generated trajectories of \eqref{eq: generative-SDE} satisfy 
\begin{equation*}
\mathbf{P}^\top\Phi_t(\mathbf{x}_0;\mathbf{L}) \overset{d}{=} \Phi_t(\mathbf{P}^\top\mathbf{x}_0;\mathbf{P}^\top\mathbf{L}\mathbf{P}),
\end{equation*}
for all $\mathbf{P} \in \mathcal{P}$ and $t \in [0,1]$, where $\smash{\overset{d}{=}}$ means equal in distribution.
\end{proposition}
\subsection{Stability}

To establish the stability of the generated distributions, we first formally introduce the structural stability and Lipschitz continuity constants characterizing the drift vector field components in \eqref{eq: generative-drift}. As stated in Section \ref{sec: stability-GNN}, graph filters and GNNs are fundamentally stable operators. For a perturbation $\|\mathbf{E}^\star\|_2 \leq \varepsilon$ and its corresponding optimal permutation $\mathbf{P}_0$ (cf. \eqref{eq: perturbation-model}), we define the structural stability constants $\Gamma_H$ and $\Gamma_u$ such that the filter and GNN in \eqref{eq: generative-drift} satisfy
\begin{equation*}
\begin{aligned}
    \|\mathbf{P}_0^\top \mathbf{H}_t(\tilde{\mathbf{L}})\mathbf{x}-\mathbf{H}_t(\mathbf{L})\mathbf{P}_0^\top \mathbf{x}\| &\leq\big(\Gamma_H\varepsilon+\mathcal{O}(\varepsilon^2)\big)\|\mathbf{x}\|, \\
    \|\mathbf{P}_0^\top u_t^{\boldsymbol{\theta}}(\mathbf{x};\tilde{\mathbf{L}})-u_t^{\boldsymbol{\theta}}(\mathbf{P}_0^\top \mathbf{x};\mathbf{L})\| &\leq\big(\Gamma_u\varepsilon+\mathcal{O}(\varepsilon^2)\big)\|g(\mathbf{x}, t)\|,
\end{aligned}
\end{equation*}
for all $\mathbf{x}\in\mathbb{R}^N$ and $t \in [0,1]$. Furthermore, graph filters and GNNs are Lipschitz continuous~\cite[Prop. 3]{schmidt_2026}. We denote their respective one-sided Lipschitz constants as $m_H$ and $m_u$, satisfying
\begin{equation*}
\begin{aligned}
    \langle \mathbf{H}_t(\mathbf{L})(\mathbf{x}-\mathbf{y}),\mathbf{x}-\mathbf{y}\rangle &\leq m_H\|\mathbf{x}-\mathbf{y}\|^2, \\
    \langle u_t^{\boldsymbol{\theta}}(\mathbf{x};\mathbf{L})-u_t^{\boldsymbol{\theta}}(\mathbf{y};\mathbf{L}),\mathbf{x}-\mathbf{y}\rangle &\leq m_u\|\mathbf{x}-\mathbf{y}\|^2,
\end{aligned}
\end{equation*}
for all $\mathbf{x},\mathbf{y}\in\mathbb{R}^N$ and $t \in [0,1]$. Combining these properties yields the stability and continuity of the drift vector field $b_t(\mathbf{x};\mathbf{L})$.

\begin{lemma}
\label{lem: comp-stability}
The vector field $b_t(\mathbf{x};\mathbf{L})$ in \eqref{eq: generative-drift} is structurally stable, i.e., for the optimal permutation $\mathbf{P}_0$ defining the perturbation one has
\begin{equation*}
\begin{aligned}
    \|\mathbf{P}_0^\top b_t(\mathbf{x};\tilde{\mathbf{L}})-b_t(\mathbf{P}_0^\top \mathbf{x};\mathbf{L})\| \leq \varepsilon \big(\Gamma_H \|\mathbf{x}\| + \Gamma_u \|g(\mathbf{x}, t)\|\big) \\
    + \mathcal{O}(\varepsilon^2)(\|\mathbf{x}\| + \|g(\mathbf{x}, t)\|).
\end{aligned}
\end{equation*}
\end{lemma}

\begin{lemma}
\label{lem: comp-lipschitz}
The vector field $b_t(\mathbf{x};\mathbf{L})$ is one-sided Lipschitz with combined constant $m_b=m_H+m_u$, i.e.,
\begin{equation*}
    \langle b_t(\mathbf{x};\mathbf{L})-b_t(\mathbf{y};\mathbf{L}),\mathbf{x}-\mathbf{y}\rangle\leq m_b\|\mathbf{x}-\mathbf{y}\|^2.
\end{equation*}
\end{lemma}

Equipped with these properties, we can establish the following key stability result for the distributions generated by the SDE \eqref{eq: generative-SDE}.

\begin{theorem}
\label{thm: wasserstein-stability}
For any initial distribution $\mathbf{x}_0\sim p_0$, assume that the components of the drift $b_t(\mathbf{x};\mathbf{L})$ in \eqref{eq: generative-drift} are permutation equivariant, and that $b_t$ satisfies the structural stability and one-sided Lipschitz properties stated in Lemmata \ref{lem: comp-stability} and \ref{lem: comp-lipschitz}. Then, the Wasserstein-2 ($W_2$) distance between the probability laws generated by the nominal and perturbed dynamics (cf. \eqref{eq: dynamics}) satisfy
\begin{equation}\label{eq:W2_stability_bound}
\min_{\mathbf{P}\in\mathcal{P}}
W_2\!\left(\tilde{\pi}_t^{\mathbf{P}},\pi_t^{\mathbf{P}}\right)
\leq
\Omega_t \big( \Gamma_H C_{H, t} + \Gamma_u C_{u, t} \big) \varepsilon
+ \mathcal{O}(\varepsilon^2),
\end{equation}
for all $t\in[0,1]$, where $\Omega_t \coloneqq \int_0^t e^{m_b(t-s)}ds$ and 
\begin{equation*}
\begin{aligned}
C_{H, t} &\coloneqq \max_{\mathbf{P} \in \mathcal{P}} \sup_{s\in[0,t]} \sqrt{\mathbb{E}\left[ \left\|\Phi_s(\mathbf{P}^\top\mathbf{x}_0;\mathbf{L})\right\|^2 \right]}, \\
C_{u, t} &\coloneqq \max_{\mathbf{P} \in \mathcal{P}} \sup_{s\in[0,t]} \sqrt{\mathbb{E}\left[ \left\|g\big(\Phi_s(\mathbf{P}^\top\mathbf{x}_0;\mathbf{L}), s\big)\right\|^2 \right]}.
\end{aligned}
\end{equation*}
\end{theorem}

\begin{proof}[Proof sketch]
Recall $\mathbf{P}_0 \in \mathcal{P}$  in \eqref{eq: perturbation-model}, let $\mathbf{x}_t \coloneqq \Phi_t(\mathbf{P}_0^\top \mathbf{x}_0; \mathbf{L})$ and $\tilde{\mathbf{x}}_t \coloneqq \Phi_t(\mathbf{x}_0; \tilde{\mathbf{L}})$, and define $\mathbf{z}_t \coloneqq \mathbf{P}_0^\top \tilde{\mathbf{x}}_t - \mathbf{x}_t$.
Consider a synchronous coupling sharing the initial state $\mathbf{x}_0$ and driving the perturbed trajectory with the standard Brownian motion $\tilde{\mathbf{w}}_t \coloneqq \mathbf{P}_0\mathbf{w}_t$. By the equivariance $\mathbf{P}_0^\top\mathbf{\Sigma}_t\mathbf{P}_0 = \mathbf{\Sigma}_t$, the stochastic noise cancels, yielding the ordinary differential equation (ODE)
\begin{equation*}
d\mathbf{z}_t = \big[ \mathbf{P}_0^\top b_t(\tilde{\mathbf{x}}_t; \tilde{\mathbf{L}}) - b_t(\mathbf{x}_t; \mathbf{L}) \big] dt.
\end{equation*}
Adding and subtracting $\mathbf{P}_0^\top b_t(\mathbf{P}_0\mathbf{x}_t; \tilde{\mathbf{L}})$, we have
\begin{equation*}
\begin{aligned}
\frac{1}{2}\frac{d}{dt}\|\mathbf{z}_t\|^2 &= \langle \mathbf{z}_t, \mathbf{P}_0^\top \big( b_t(\tilde{\mathbf{x}}_t; \tilde{\mathbf{L}}) - b_t(\mathbf{P}_0\mathbf{x}_t; \tilde{\mathbf{L}}) \big) \rangle \\
& \quad + \langle \mathbf{z}_t, \mathbf{P}_0^\top b_t(\mathbf{P}_0\mathbf{x}_t; \tilde{\mathbf{L}}) - b_t(\mathbf{x}_t; \mathbf{L}) \rangle.
\end{aligned}
\end{equation*}
Applying Lemma \ref{lem: comp-lipschitz} to the first term, and Cauchy-Schwarz with Lemma \ref{lem: comp-stability} to the second, yields
\begin{equation*}
\begin{aligned}
\frac{1}{2}\frac{d}{dt}\|\mathbf{z}_t\|^2 \leq m_b \|\mathbf{z}_t\|^2 + \Big( \varepsilon \big(\Gamma_H \|\mathbf{x}_t\| + \Gamma_u \|g(\mathbf{x}_t, t)\|\big) \\
+ \mathcal{O}(\varepsilon^2)\big(\|\mathbf{x}_t\| + \|g(\mathbf{x}_t, t)\|\big) \Big) \|\mathbf{z}_t\| .
\end{aligned}
\end{equation*}
Dividing by $\|\mathbf{z}_t\|$ and applying Gr\"onwall's inequality ($\mathbf{z}_0 = \mathbf{0}$) gives
\begin{equation*}
\begin{aligned}
\|\mathbf{z}_t\| \leq \int_0^t e^{m_b(t-s)} \Big( \varepsilon \big(\Gamma_H \|\mathbf{x}_s\| + \Gamma_u \|g(\mathbf{x}_s, s)\|\big) \\
+ \mathcal{O}(\varepsilon^2)\big(\|\mathbf{x}_s\| + \|g(\mathbf{x}_s, s)\|\big) \Big) ds. 
\end{aligned}
\end{equation*}
The 2-Wasserstein distance is bounded by this coupling, namely
$W_2(\tilde{\pi}_t^{\mathbf{P}_0}, \pi_t^{\mathbf{P}_0}) \leq \sqrt{\mathbb{E}\big[\|\mathbf{z}_t\|^2\big]}$. Applying Minkowski's integral inequality and the $L^2$ triangle inequality, yields
\begin{equation*}
\begin{aligned}
& \sqrt{\mathbb{E}\big[\|\mathbf{z}_t\|^2\big]}
\leq 
\\
& 
\varepsilon
\int_0^t e^{m_b(t-s)}
\Big(
\Gamma_H \sqrt{\mathbb{E}\big[\|\mathbf{x}_s\|^2\big]}
+
\Gamma_u \sqrt{\mathbb{E}\big[\|g(\mathbf{x}_s,s)\|^2\big]}
\Big)ds
\\
&
+
\mathcal{O}(\varepsilon^2)
\int_0^t e^{m_b(t-s)}
\Big(
\sqrt{\mathbb{E}\big[\|\mathbf{x}_s\|^2\big]}
+
\sqrt{\mathbb{E}\big[\|g(\mathbf{x}_s,s)\|^2\big]}
\Big)ds .
\end{aligned}
\end{equation*}
Using $C_{H,t} \geq \sqrt{\mathbb{E}[\|\mathbf{x}_s\|^2]}$ and $C_{u,t} \geq \sqrt{\mathbb{E}[\|g(\mathbf{x}_s,s)\|^2]}$ yields the bound in \eqref{eq:W2_stability_bound}, completing the proof sketch.
%
%
\end{proof}

We find stability is governed by three terms: (i) $\Gamma_H$ and $\Gamma_u$, which capture the inherent stability of the graph filter and GNN; (ii) $\Omega_t$, which is exponential in the combined Lipschitz constant and dictates how the continuous dynamics amplify errors; and (iii) $C_{H, t}$ and $C_{u, t}$, the trajectory suprema across all initial condition permutations.

\begin{figure*}[!t]
\centering
\includegraphics[width=0.45\textwidth]{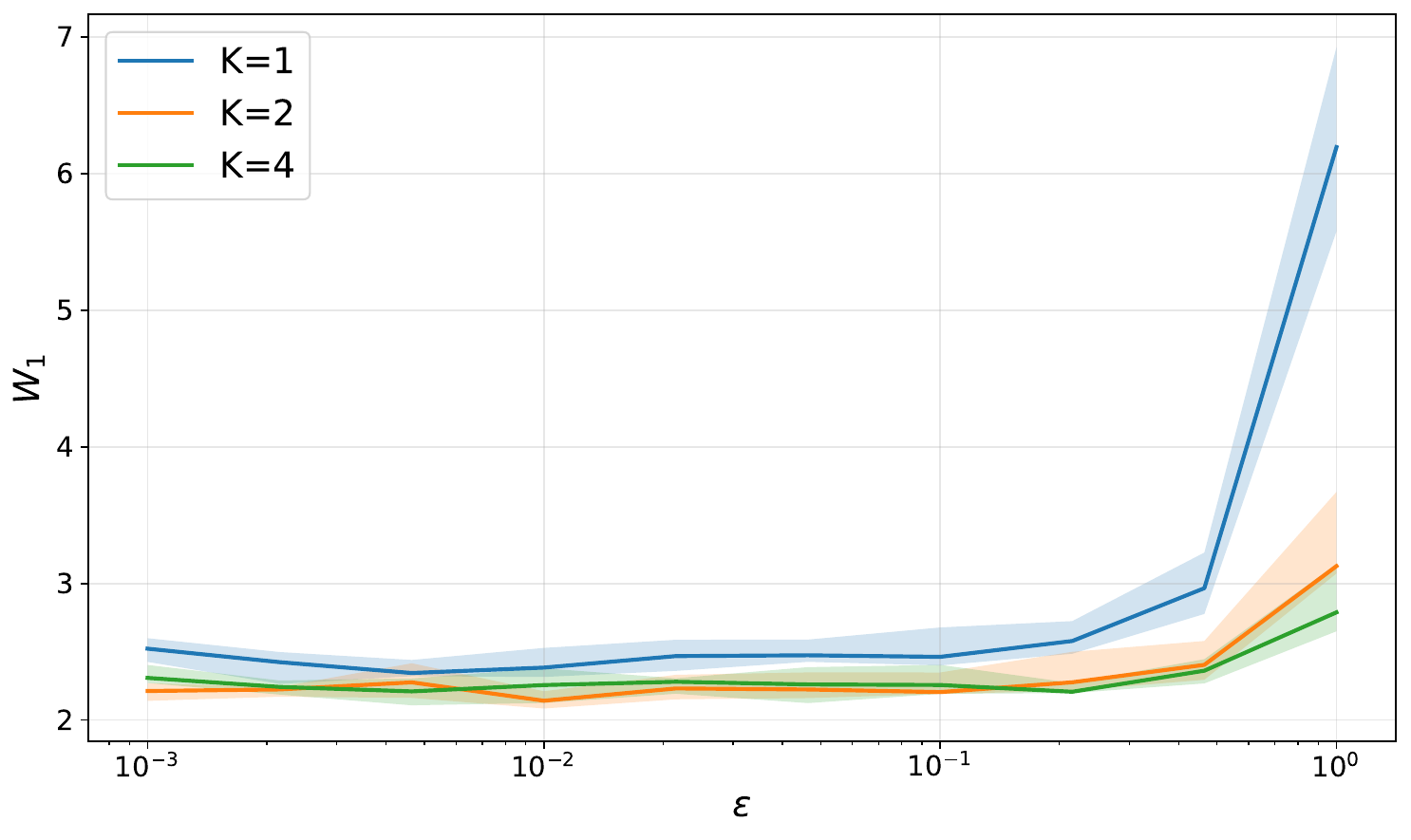}
\hfill
\includegraphics[width=0.45\textwidth]{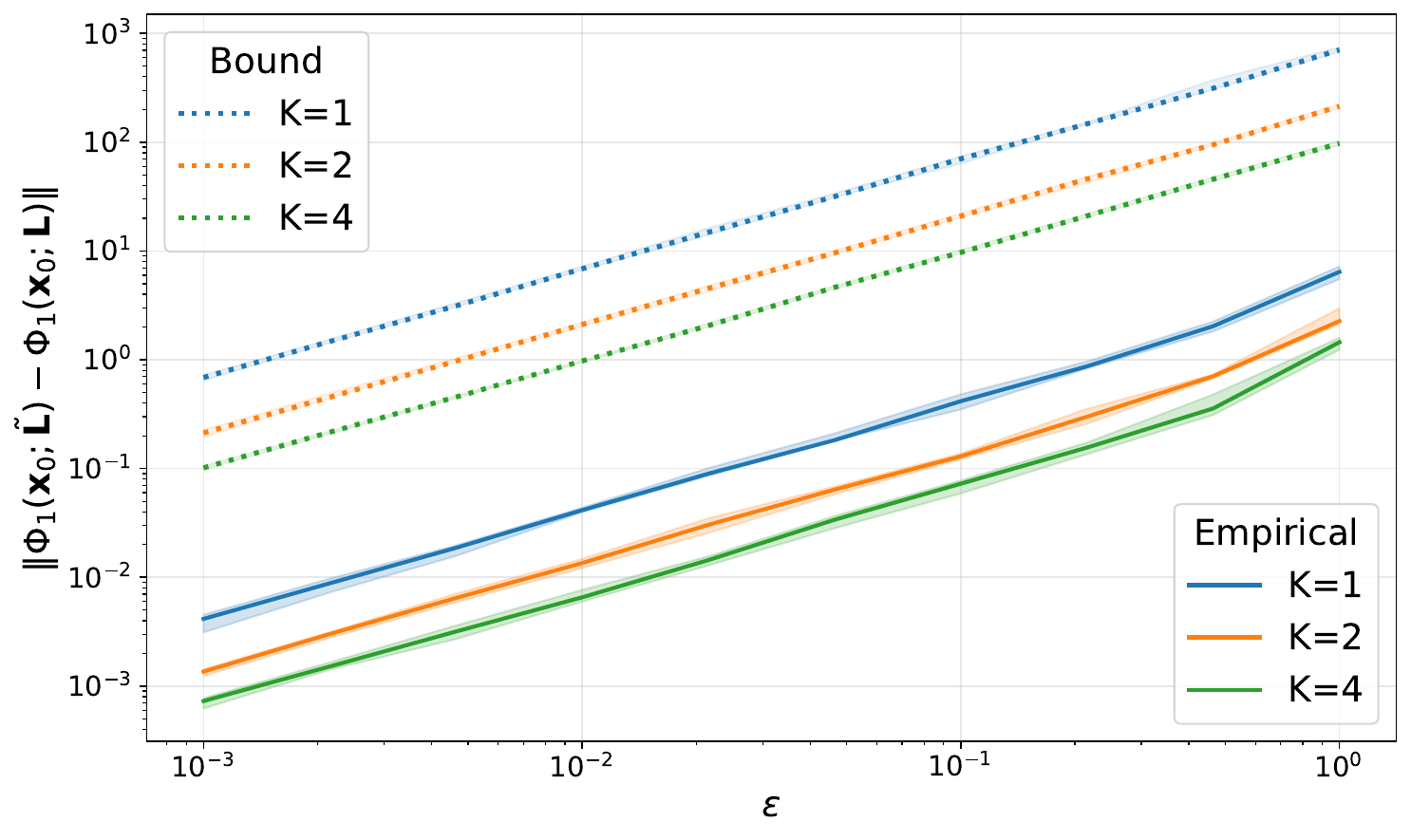}

\vspace{3pt} 

\includegraphics[width=0.45\textwidth]{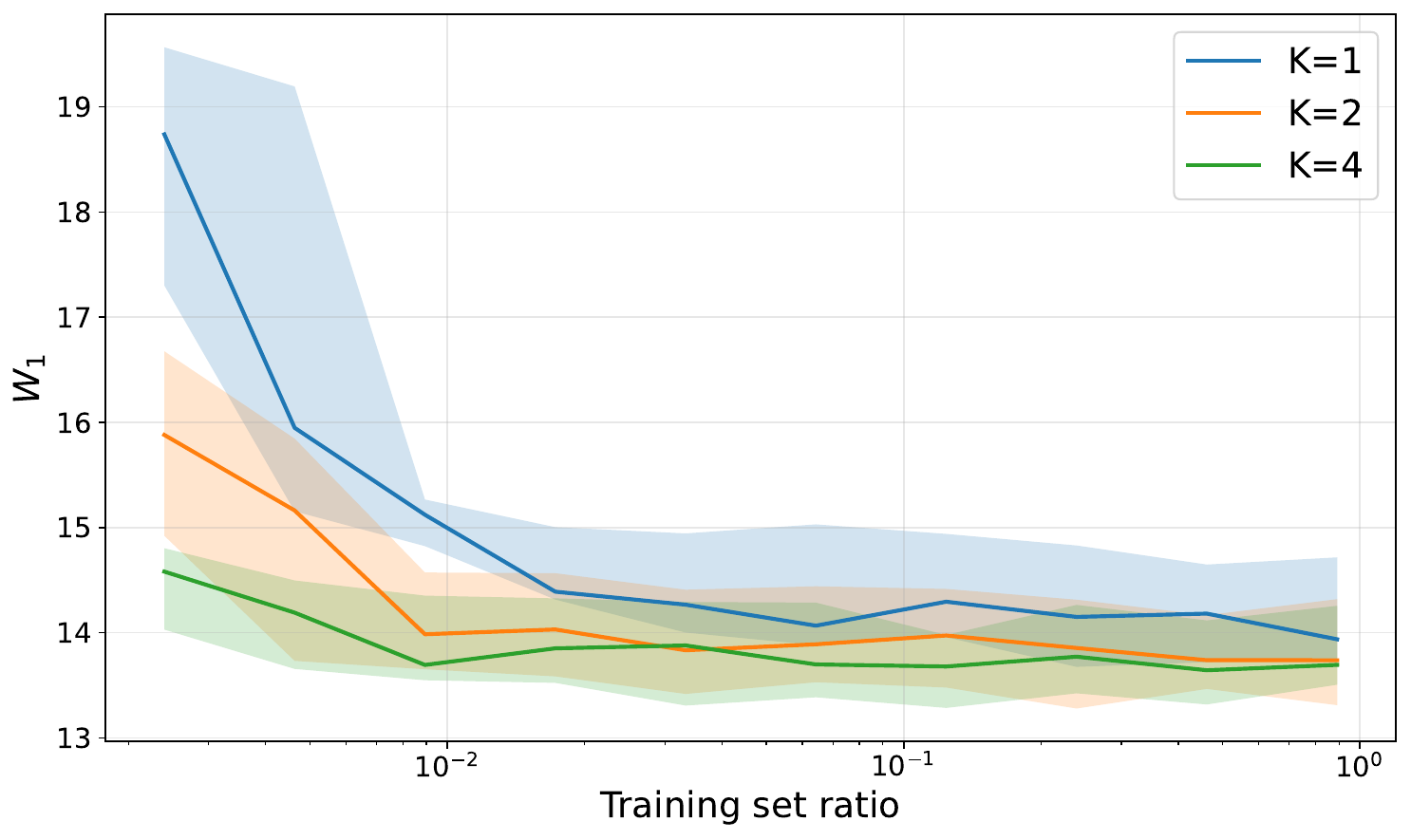}
\hfill
\includegraphics[width=0.45\textwidth]{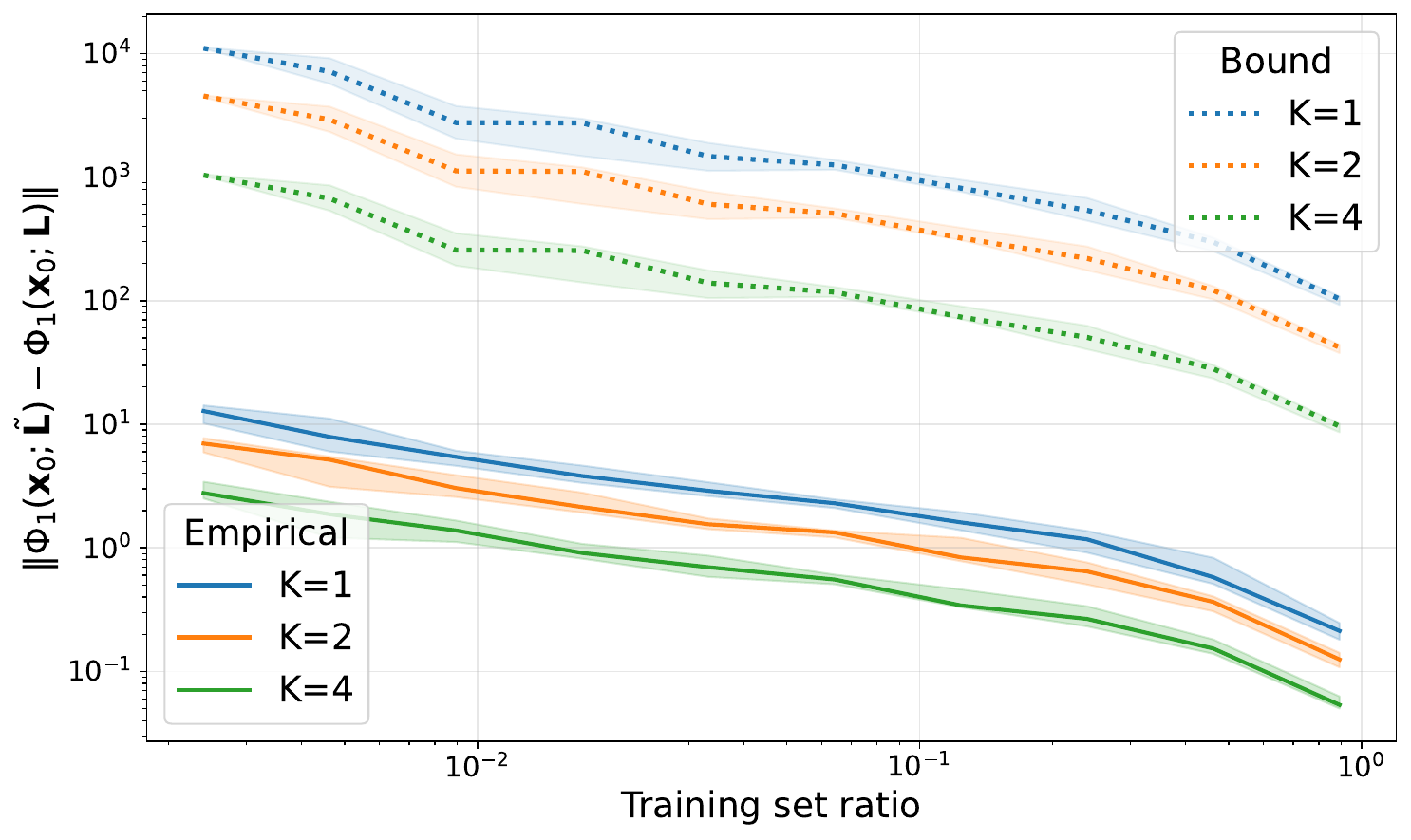}

\vspace{-5pt} 
\caption{Performance and stability under structural graph perturbations. Top: SBM graph under controlled synthetic relative perturbations ($\varepsilon$). Bottom: fMRI signals under data-driven graph estimation errors, varying the fraction of training samples used. Generative quality (left) shows higher-order optimized filters ($K=2, 4$) achieve competitive or better $W_1$ distances compared to the $K=1$ baseline (heat equation). Empirical stability (right) demonstrates that increasing the optimized filter order yields significantly lower output variation under perturbations. Lines and shaded regions represent medians and 25th–75th percentiles across 10 independent runs.}
\label{fig:all_results}
\vspace{-5pt} 
\end{figure*}

\section{Stable Graph Filters}\label{sec:stable_filters}

Theorem \ref{thm: wasserstein-stability} offers actionable insights to promote robustness in the generative pipeline. While the GNN constants $\{\Gamma_u, m_u\}$ can be regularized during training~\cite[Sec. V]{schmidt_2026}, the filter constants $\{\Gamma_H, m_H\}$ are determined by the choice of reference dynamics. Instead of relying on the standard heat equation whereby $\mathbf{H}(\mathbf{L}) = -\kappa \mathbf{L}$ in \eqref{eq: generative-drift}, here we formulate a principled filter design problem that enforces similar signal smoothness while improving its structural stability.

We henceforth assume that $\mathbf{L}$ is a normalized graph Laplacian with spectrum in $[0,\lambda_{\max}]$ and restrict our attention to time-independent reference filters by setting $h_k(t)\equiv h_k$ for $k=0,\ldots,K$, so that $\mathbf{H}(\mathbf{L})\coloneqq\sum_{k=0}^{K}h_k\mathbf{L}^k$. Nulling the bias term, $h_0=0$, we parameterize the corresponding frequency response as $h(\lambda)\coloneqq\sum_{k=1}^{K}h_k\lambda^k$, with $\mathbf{h}\coloneqq[h_1,\ldots,h_K]^\top\in\mathbb{R}^K$. Further imposing no signal amplification, $h(\lambda)\leq 0$, yields a one-sided Lipschitz constant $m_H=0$. Thus, enhancing robustness reduces to minimizing the structural stability constant $\Gamma_H$. Since $\Gamma_H$ is proportional to the integral Lipschitz constant of the filter \cite{Gama_2020}, our design objective is to select $\mathbf{h}$ so as to minimize $\max_{\lambda}|\lambda h'(\lambda)|$.

To also guarantee the filter generates smooth signals, we bound the Dirichlet energy $\mathcal{E}(\mathbf{x}) = \mathbf{x}^\top \mathbf{L} \mathbf{x}$ of its output. Under the linear dynamics $d\mathbf{x}_t = \mathbf{H}(\mathbf{L})\mathbf{x}_t dt$, the spectral response at $t=1$ and frequency $\lambda_i$ is $\hat{x}_{1,i} = e^{h(\lambda_i)}\hat{x}_{0,i}$. For a normalized initial signal $\|\mathbf{x}_0\|\leq 1$, imposing a prescribed smoothness bound $\eta > 0$ yields
\begin{equation*}
    \mathcal{E}(\mathbf{x}_1) \leq \max_{\lambda \in [0, \lambda_{\text{max}}]} \lambda e^{2h(\lambda)} \leq \eta.
\end{equation*}
Taking the logarithm directly simplifies this smoothness constraint to $h(\lambda) \leq \frac{1}{2}\ln(\frac{\eta}{\lambda})$, for all $\lambda > 0$. 

Introducing an auxiliary variable $\rho \geq 0$, we formulate the min-max filter design problem in epigraph form as
\begin{equation}
\label{eq: filter-silp}
\begin{aligned}
    & \min_{\mathbf{h}, \rho} \rho \\
    \text{s.t.} \quad
    & -\rho \leq \lambda h'(\lambda) \leq \rho, \qquad &&\forall \lambda \in (0, \lambda_{\text{max}}],\\
    & h(\lambda) \leq \frac{1}{2} \ln\left(\frac{\eta}{\lambda}\right), \qquad &&\forall \lambda \in (0, \lambda_{\text{max}}], \\
    & h(\lambda) \leq 0, \qquad &&\forall \lambda \in [0, \lambda_{\text{max}}].
\end{aligned}
\end{equation}
Since \(h(\lambda)\) and \(h'(\lambda)\) are linear in the polynomial coefficients $\mathbf{h}$, \eqref{eq: filter-silp} is a semi-infinite linear program. In practice, it can be approximated by discretizing \((0,\lambda_{\text{max}}]\) over a dense grid and solving the resulting finite-dimensional linear program \cite[Sec. 7]{hettich1993semi}.

\section{Numerical Experiments}\label{sec:experiments}

Here we test the proposed stable filters on a graph signal generation task using samples from a target data distribution. We consider two test cases: (i) a synthetic setting based on a Stochastic Block Model (SBM) graph \cite{holland1983stochastic}; and (ii) a real-world scenario using fMRI data \cite{HCP}, where functional brain connectivity defines the graph structure. Since our prime objective is to perform an ablation study on the graph filter design, we utilize a fixed architecture and training procedure for the learned vector field $u_t^{\boldsymbol{\theta}}$ across all configurations. For an approach to robustify $u_t^{\boldsymbol{\theta}}$ itself; see~\cite{schmidt_2026}. Code to reproduce the experiments is available at \href{https://github.com/mschmi21/stable_graph_fm}{github.com/mschmi21/stable\_graph\_fm}.

\subsection{Experimental Setup}

\noindent\textbf{Datasets.} In the synthetic setting, we use an SBM graph with two communities of 10 nodes each ($N=20$). Graph signals are drawn from a Gaussian distribution with standard deviation $1$. One community has a mean of $1$, while the other has a mean of $-1$. For the real-data test case, we use a single subject from the HCP dataset, producing data in $\mathbb{R}^{360 \times 1190}$ ($N=360$ brain regions, $1190$ time points). Each time point is treated as an individual graph signal.\vspace{2pt}

\noindent\textbf{Filter Design and Architecture.} To isolate the effect of the reference dynamics, we synthesize stable filters of orders $K \in \{1, 2, 4\}$ by solving the semi-infinite linear program in \eqref{eq: filter-silp}. Notably, for $K=1$, this optimization naturally recovers the heat equation filter, which serves as our baseline. The smoothness constraint parameter $\eta$ is selected via a grid search on the validation set. For the learned vector field $u_t^{\boldsymbol{\theta}}$, we employ the \texttt{GCNPolicy} model introduced in \cite{yang2025topological}. For simplicity, we henceforth set $\boldsymbol{\alpha}_t=\mathbf{0}$ and $\mathbf{\Sigma}_t=\mathbf{0}$ for all $t\in[0,1]$. \vspace{2pt}

\noindent\textbf{Perturbations and Evaluation.} We train all models on the unperturbed, nominal graph, using an independent coupling and filter-dependent optimal paths (cf. \cite{wyrwal2026topological}). Robustness is evaluated at test time. For the SBM graphs, we introduce a controlled synthetic relative perturbation of the form $\tilde{\mathbf{L}} = \mathbf{L} + \frac{1}{2}(\mathbf{E}\mathbf{L} + \mathbf{L}\mathbf{E})$, with $\|\mathbf{E}\|_2 = \varepsilon$. For the fMRI data, perturbations are data-driven: test-time graphs are constructed using empirical correlation matrices estimated from varying data subsets. We report two metrics: (i) the Wasserstein-1 ($W_1$) distance to quantify generative performance; and (ii) the variation $\|\Phi_1(\mathbf{x}_0; \tilde{\mathbf{L}}) - \Phi_1(\mathbf{x}_0; \mathbf{L})\|$ in generated outputs induced by graph perturbations, estimated via an Euler sampler.

\subsection{Results and Discussion}

Fig.~\ref{fig:all_results} depicts the results for the synthetic SBM (top) and fMRI (bottom) experiments. In terms of generative quality measured by the $W_1$ distance, the higher-order stable filters ($K=2, 4$) achieve competitive or slightly better performance compared to the first-order baseline ($K=1$) under low perturbations, while clearly outperforming it under larger perturbations. More importantly, the empirical stability evaluations demonstrate that increasing the filter order offers added robustness to structural errors. By actively minimizing the structural stability constant $\Gamma_H$, the generative trajectories of the higher-order filters exhibit markedly lower output divergence. Finally, while the theoretical stability bounds conservatively overestimate the empirical errors, they remain highly informative and align with the relative differences observed in practice.

\section{Conclusion}\label{sec:conclusion}
In this work, we established explicit stability bounds for a general class of graph-aware continuous-time generative dynamics (cf. Theorem \ref{thm: wasserstein-stability}). By quantifying the Wasserstein distance between distributions generated under nominal and perturbed graphs, our analysis reveals how the choice of reference dynamics directly impacts the structural robustness of the model. Guided by these theoretical insights, we introduced a principled optimization framework to design stable graph filters. Our empirical evaluations confirm that appropriately selecting this reference filter significantly enhances the stability of the generative model against graph perturbations, without sacrificing the quality of the generated signals.

\section{Compliance with Ethical Standards}

This study uses previously collected, de-identified Human Connectome Project (HCP) data~\cite{HCP}. No new human-subject data were collected, and no additional ethical approval was required.

\bibliographystyle{IEEEtran}
\bibliography{refs}

\end{document}